\documentclass[preprint,12pt]{article}
\usepackage[utf8]{inputenc}
\usepackage{fancybox, calc}  
\usepackage{amssymb}
\usepackage{amsmath}
\usepackage{hyperref}
\usepackage{graphicx}
\usepackage{enumitem}
\usepackage{ucs}
\usepackage{amsfonts}
\usepackage{amssymb}
\usepackage{makeidx}
\usepackage{diagbox}
\usepackage{tabularx}
\usepackage{float}
\usepackage[usenames, dvipsnames]{color}
\usepackage{url}
\usepackage{ marvosym }
\usepackage{ wasysym }
\usepackage{ mathrsfs }
\usepackage[utf8]{inputenc}
\usepackage[english]{babel}
\usepackage[normalem]{ulem}
\usepackage{ucs}
\usepackage{amsfonts}
\usepackage{dcolumn}% Align table columns on decimal point
\usepackage{makeidx}
\usepackage{bm}% bold math
\usepackage[usenames]{color}
\usepackage{multirow}
\usepackage{graphicx}
\usepackage{hyperref}
\usepackage{subfig}
\usepackage{booktabs}
\usepackage{xcolor}
\usepackage{amsmath, amsthm, amssymb} 
\usepackage{phaistos}
\usepackage{soul}
\restylefloat{table}

\newtheorem{theorem}{Theorem}

\newtheorem{lemma}{Lemma}
\newtheorem{definition}{Definition}
\newtheorem{proposition}{Proposition}

\newtheorem{property}{Property}

\title{Differential-escort transformations and the monotonicity of the LMC-Rényi complexity measure}
\author{D. Puertas-Centeno\\ \footnotesize
	Departamento de F\'{\i}sica At\'{o}mica, Molecular y Nuclear, and \\ \footnotesize
	Instituto Carlos I de F\'{\i}sica Te\'orica y Computacional,\\\footnotesize Universidad de Granada, Granada 18071, Spain}

\begin{document}
\maketitle

%Escort distributions have been shown to be very useful in a great variety of fields ranging from information theory, nonextensive statistical mechanics till coding theory, chaos and multifractals. In this work we give the notion and the properties of a novel type of escort distributions, the \textit{differential escort} distributions, which have various advantages with respect to the standard ones \textcolor{red}{(*)}. We highlight the behavior of the Shannon, Rényi and Tsallis entropies of these distributions. Then, we illustrate their utility to prove the monotonicity property of the LMC-Rényi complexity measure and to obtain the Tsallis q-exponential densities as the \textit{differential escort} transformation of the exponential density. Finally, this transformation allows us to study the behavior of general N-piecewise distributions in the two extreme cases of minimal and very high LMC-Rényi complexity. 

Escort distributions have been shown to be very useful in a great variety of fields ranging from information theory, nonextensive statistical mechanics till	 coding theory, chaos and multifractals. In this work we give the notion and the properties of a novel type of escort density, the \textit{differential-escort densities}, which have various advantages with respect to the standard ones. We highlight the behavior of the differential Shannon, Rényi and Tsallis   entropies of these distributions. Then, we illustrate their utility to prove the monotonicity property of the LMC-Rényi complexity measure and to study the behavior of general distributions in the two extreme cases of minimal and very high LMC-Rényi complexity. Finally, this transformation allows us to obtain the Tsallis q-exponential densities as the differential-escort transformation of the exponential density. 

\section{Introduction}

The study of chaotic and complex systems have needed the development of mathematical tools able to capture the fundamental statistical properties of the system. Escort distributions  have been introduced in statistical physics for the characterization of multifractals systems \cite{Beck93}. These distributions $\{\tilde p_i\}$ conform a one-parameter class of transformations of an original probability distribution $\{p_i\}$ according to
$
\tilde p_i=\frac{p_i^q}{\sum_{i=1}^N p_i^q},
$ with  $q\in\mathbb R$.

This idea previously appeared in relation to the Rényi-entropy-based coding theorem \cite{Campbell65,Bercher09} and Rényi-entropy-based fractal dimensions \cite{Chhabra89}. The mathematical properties of the discrete escort distributions have been widely studied \cite{Abe03,Jizba04,Bercher12,Bercher13}. This concept can be easily extended to the continuous case. Given a real variable $x\in \mathbb R $ and a probability distribution $\rho(x)$, such that $\int_{\mathbb R} \rho(x)\,dx=1$, one has the escort distribution  \cite{Bercher10} defined as 
 \begin{equation}\label{escort}
E_q[\rho](x)\equiv\tilde\rho(x)=\frac{[\rho(x)]^q}{\int_{\mathbb R}[\rho(t)]^q\,dt},
 \end{equation}
on the assumption that $\int_{\mathbb R} \rho(x)^q\,dx<\infty$. Note that the parameter $q$ plays a focus role to highlight different regions of $\rho(x)$. These distributions play a relevant role in coding problems, non-equilibrium statistical mechanics \cite{Abe02,Tsallis04} and electronic structure \cite{sheila08,sheila09,vignat12}. A particular example is the q-exponential distribution 
\begin{equation}\label{q-exponential}
e_q(x)\propto (1+(q-1)|x|)^{\frac1{1-q}}
\end{equation}
which maximizes the  differential Rényi entropy
\begin{equation}\label{Renyi}
R_q[\rho]=\frac1{1-q}\log\left(\int_{\mathbb R} [\rho(x)]^q\,dx\right),
\end{equation}
and the differential  Tsallis entropy
\begin{equation}\label{Tsallis}
T_q[\rho]=\frac{1}{1-q}\left(1-\int_{\mathbb R}[\rho(x)]^q\,dx\right),
\end{equation}
subject to averarage-constraints governed by its escort distribution.
Of course, in the limit $q\to1$ the original distribution is recovered in Eq. \eqref{escort}, the exponential distribution is also recovered in Eq. \eqref{q-exponential} and the  differential Shannon   entropy 
\begin{equation}\label{Shannon}
S[\rho]=\lim_{q\to1}R_q[\rho]=\lim_{q\to1}T_q[\rho]=-\int_{\mathbb R} \rho(t)\log[\rho(t)]\,dt
\end{equation}
is respectively recovered  in Eqs \eqref{Renyi} and \eqref{Tsallis}.

The aim of this work is to introduce the notion of \textit{differential-escort transformation}, $\mathfrak E_\alpha$, and to study its basic mathematical properties (probability invariance, composition rule, scaling property,...). 
Then, we highlight the strongly regular behavior of the differential Shannon, Rényi and Tsallis   entropies under this transformation, observing that the entropic parameter naturally rescales similarly to the rescaling behavior recently found  by Korbel \cite{Korbel17} for the non-additivity parameter in Tsallis thermostatistics \cite{Tsallis88}. This behavior is related to the rescaling of the relative fluctuations of a system with a finite number of particles, and plays a relevant role in the deformed calculus developed by Borges \cite{Borges04} as it is discussed by Korbel himself. Moreover we also note that the q-exponential distribution is just the differential-escort transformation of the standard exponential distribution; so, differently to what happens to the the standard escort transformation of an exponential distribution which is another exponential. In fact, we show that the differential-escort transformation changes the behavior of the distribution tail in a deeper and more interesting manner than the standard escort transformation, which allows us to propose a possible characterization of power-law-decaying probability densities through Lemma \ref{lemma}.  On the other hand, we carry out a study of the behaviour of the LMC-Rényi complexity measure \cite{Lopez-Ruiz09,Sanchez-Moreno14} of the transformed densities.  Actually, the notion of \textit{differential-escort density} allows us to solve the  monotonicity problem of the LMC-Rényi complexity measure recently posed by Rudnicki et al \cite{Rudnicki16}. This fact, in turn, permits us to define the notions of \textit{low-complexity }and \textit{high-complexity }probability densities in the LMC-Rényi sense, as well as we characterize the entropic behavior of these probability densities.

 The structure of this work is the following: In section \ref{sec:Diff-esc} the differential- escort transformation  $\mathfrak E_\alpha$ is  defined and its basic mathematical properties are given. In section \ref{sec:Entropic} the entropic properties of the differential-escort densities are discussed. In section \ref{sec:LMC-R} the LMC-Rényi complexity of the differential-escort densities is studied and the monotonicity property of this measure is proven. In section \ref{sec:extreme} the entropic and  complexity behavior of a general probability density when it is deformed until to the low and high complexity limits is studied. Then, in section \ref{sec:q-exp} this transformation is applied to distributions of exponential, q-exponential and  general power-law decaying.
 Finally, in section \ref{sec:conc} some conclusions and open problems are given.

\section{Differential-escort transformation} \label{sec:Diff-esc}

In this section we give the notion and properties of the differential-escort transformation. Let us advance that the basic difference with the standard escort transformations is the normalization process. Indeed, while an escort density is normalized according to \eqref{escort}, in the differential-escort case the normalization is achieved through a variable change which imposes the conservation of the probability in any differential interval of the support, as we will see later.

\subsection*{The notion}
 
	Let us consider a probability density $\rho(x), x\in \Lambda \subseteq\mathbb R$, normalized, so that $\int_{\Lambda}\rho(x)dx = 1$; and let us denote $\mathcal D(\mathbb R)$ for the set of any distribution $\rho$ on any subset of $\mathbb R$.
        
\begin{definition}\label{definition}

         Let $\alpha \in\mathbb{R}$, and $\rho\in \mathcal D(\mathbb R)$ be a probability density with a connected support \footnote{We assume that the support is connected for easier reading. The definition could be easily extended to any distribution without disturbing its properties.} $\Lambda$. We define the transformation $\mathfrak E_\alpha:\mathcal D(\mathbb R)\longrightarrow \mathcal D(\mathbb R) $ as:
        
 \begin{equation}\label{def-E}
 \mathfrak {E}_\alpha[\rho](y)\equiv [\rho(x(y))]^\alpha,
 \end{equation}
 where $y=y(x)$ is a bijection defined by:

\begin{equation} \label{def-y}
\frac{dy}{dx}=[\rho(x)]^{1-\alpha},\quad y(x_0)=x_0,\,\,x_0\in\Lambda.
 \end{equation}
 The support $\Lambda_\alpha$ of the transformed density $\mathfrak E_\alpha[\rho]$ is given as $\Lambda_\alpha=y(\Lambda)=\{y\in\mathbb R\,\arrowvert \,y=y(x),x\in\Lambda\}$. To make an easier reading we will denote $ \rho_\alpha(y)\equiv \mathfrak E_\alpha[\rho](y)$, and generally we will take $x_0=0$, and $y(x)=\int_0^x [\rho(t)]^{1-\alpha}\,dt.$
\end{definition}

%For the sake of a easier reading we take the following notation $\rho_\alpha(y)=\mathfrak {E}_\alpha[\rho](y)$. 

We remark that this definition is valid for any $\alpha\in\mathbb R$, contrary with the standard escort distribution for which  the parameter $\alpha$ is restricted by the condition  $\int_{\Lambda}[\rho(x)]^\alpha\,dx<\infty$ as already indicated in Eq. (\ref{escort}). This extension is possible since the support of a differential-escort density $\Lambda_\alpha$ does not remain invariant contrary to the standard escort case. 
As we will see later, for any probability density $\rho$, the operation (\ref{definition}) defines a transformed density $\rho_\alpha$ for any $\alpha\in\mathbb R$. 
%In general, for easier reading, we will take $x_0=0$, and so $y(x)=\int_0^x [\rho(t)]^{1-\alpha}\,dt.$

Let us also point out that the selection of $x_0$ only implies a translation. In addition, when $\alpha=1$, one has that the operation  $\mathfrak E_\alpha$ corresponds to the identity, i.e., $\mathfrak E_1 [\rho]=\rho$; and when  $\alpha=0$, the operation $\mathfrak E_\alpha$ transforms $\rho$ to an uniform distribution with an unitary support, concretely $$\mathfrak E_0[\rho](x)=\left\{ \begin{array} {c} 1,\quad x\in[x_0\,-\,p_-,\,x_0\,+p_+] \\\hspace{-2.5cm} 0,\quad \text{otherwise}
\end{array}  \right.,$$
where $p_-=Prob[x<x_0]$ and $p_+=Prob[x>x_0] .$
\subsection*{The basic properties}

In the following we will give some basic properties of the differential-escort transformation.

Property \ref{diff-inv} is the most characteristic property of this transformation which consists in a strong probability invariance far beyond the mere conservation of the norm of the \textit{standard escort} case.
  
\begin{property}\label{diff-inv}
\textbf{Probability invariance}
\\
%\textcolor{red}{Un poco inconsistente: decimos que se satisface una propiedad que no se ha definido previamente.}\\
 Let $\alpha \in\mathbb{R}$ and $\rho$ be a probability density with a connected support $\Lambda$. Then, for any pair of points $x_1,x_2\in\Lambda$ and respectively $y_1=y(x_1)$ and $y_2=y(x_2)$ the identity
	\begin{equation}
\int_{x_1}^{x_2} \rho(x)dx=\int_{y_1}^{y_2} \rho_\alpha(y)dy,
	\end{equation}
or equivalently
\begin{equation}
Prob[x\in[x_1,x_2]]=Prob[y\in[y_1,y_2]],
\end{equation}
is fulfilled.
\begin{proof}
This property follows straightforwardly from (\ref{definition}) since
	\begin{equation*}
\int_{y_1}^{y_2} \rho_\alpha(y)\,dy=\int_{x_1}^{x_2} [\rho(x)]^{\alpha}\,\frac{dy}{dx}\,dx=\int_{x_1}^{x_2} [\rho(x)]^{\alpha}\,[\rho(x)]^{1-\alpha}\,dx=\int_{x_1}^{x_2} \rho(x)\,dx.
	\end{equation*}
\end{proof}
\end{property} 
This property makes a deep difference with escort distributions. While for the latter ones the conservation of the norm is  imposed dividing by a real number as indicated in (\ref{escort}), for the differential-escort distributions it naturally holds as a consequence of property \ref{diff-inv} since $\int_{\Lambda_\alpha}\rho_\alpha(y)\,dy=\int_\Lambda\rho(x)\,dx=1$. Moreover, a  similar property is fulfilled by a relevant transformation between auxiliary and physical probability densities in the context of quantum gravity \cite{Rastegin17}.

\begin{property}\label{compo-prop}
\textbf{Composition law}\\
	Let $\alpha$, $\alpha'$ $\in\mathbb R$,  then 
	\begin{equation}
\mathfrak E_{\alpha}[\mathfrak E_{\alpha'}[\rho]]=\mathfrak E_{\alpha'}[\mathfrak E_{\alpha}[\rho]]=\mathfrak E_{\alpha\alpha'}[\rho] 
	\end{equation}
	holds.
 
\begin{proof}
 By definition, $\mathfrak {E}_\alpha[\rho(x)](y)\equiv{\rho_\alpha}(y)=[\rho(x)]^\alpha, \, \, \,  $
              where $dy=[\rho(x)]^{1-\alpha}dx$.  Moreover,  one has
              $\mathfrak {E}_\gamma[\mathfrak {E}_\alpha[\rho(x)]](z)=\mathfrak {E}_\gamma[\rho_\alpha(y)](z)=[\rho_\alpha(y)]^\gamma=[\rho(x)]^{\alpha\cdot\gamma} $
                    where $dz=[\rho_\alpha(y)]^{1-\gamma}dy$, so that one has $dz=[\rho(x)]^{\alpha(1-\gamma)}[\rho(x)]^{1-\alpha}dx=[\rho(x)]^{1-\alpha\cdot\gamma}dx$.
\end{proof}
\end{property} 
This property is similar to the one of the standard escort transformations, but the latter one holds in a more restrictive sense by taking into account that the standard escort transformations are not typically well defined for any $\alpha\in\mathbb R$ . On the other hand this property allows us to find the inverse element of the differential-escort transformation
   \begin{equation}
  [\mathfrak E_\alpha]^{-1}=\mathfrak E_{\alpha^{-1}},\quad \alpha\not=0,
  \end{equation}
what allows us to say that $\mathfrak E_{\alpha\not=0}[\mathcal D(\mathbb R)]=\mathcal D(\mathbb R)$.

Let us finally give the composition rule between the differential-escort and the scaling transformations. For example, in the standard escort case defined in Eq. \eqref{escort},  the composition rule with the scaling transformation is given by $ E_\alpha[\rho^{(a)}]=E_\alpha[\rho]^{(a)}$, where  $\rho^{(a)}$ denotes the scaling transformed distribution $\rho^{(a)}(x)=a\rho(ax), a>0$. As stated in the following property, in the composition law for the differential-escort case the power operation is also inherited by the scaling parameter $a$.

\begin{property}
\textbf{Scaling Property}\\
Let $a \in\mathbb{R_+}$, $\alpha \in\mathbb{R}$, $\rho$ be a probability distribution, and the scaling transformed distribution $\rho^{(a)}(x)=a\rho(ax)$. Then, it holds
\begin{equation}
\mathfrak E_\alpha[\rho^{(a)}]=\mathfrak E_\alpha[\rho]^{(a^\alpha)}.
\end{equation}
  \begin{proof}
 From the hypotheses of this statement one has the associated differential-escort  distribution $\rho_\alpha(y)=\rho(x)^\alpha$, where $y(x)=\int_0^x[ \rho(t)]^{1-\alpha}dt$, or equivalently $y=\int_0^{x(y)}[ \rho(t)]^{1-\alpha}dt$. Later, we consider the differential-escort distribution of the scaling transformed, $(\rho^{(a)})_\alpha (z)=[\rho^{(a)}(x)]^\alpha$, with $z(x)=\int_0^x [\rho^{(a)}(t)]^{1-\alpha}dt$, so we have:
         $$ (\rho^{(a)})_\alpha (z)=a^\alpha [\rho(a x(z))]^\alpha, \quad z(x)=a^{-\alpha}\int_0^{ax} [\rho(t)]^{1-\alpha}dt.$$
Then, we can write $a^{\alpha}z=\int_0^{ax(z)} [\rho(t)]^{1-\alpha}dt.$ On the other hand, taking into account that $y=\int_0^{x(y)}[ \rho(t)]^{1-\alpha}dt$, we have that $ ax(z)=x(a^\alpha z)$ and finally 
\begin{equation*}
(\rho^{(a)})_\alpha (z)=a^\alpha [\rho( x(a^\alpha z))]^\alpha=a^\alpha\rho_\alpha(a^\alpha z)=\rho_\alpha^{(a^\alpha)}(z).
\end{equation*} 
     \end{proof}
\end{property}

\section{The entropic properties}\label{sec:Entropic}
  
The functional ingredients of differential Rényi and Tsallis entropies \eqref{Renyi}, \eqref{Tsallis} are the entropic moments of the probability distribution $\rho$,
\begin{equation}
W_q[\rho]=\int_{\mathbb R} [\rho(x)]^q\,dx.
\end{equation}   
In this section we will study the behavior of these entropy-like functionals for the differential-escort distributions, finding that it is much simpler than the corresponding one for the standard escort case. 
%In a recent work \cite{Korbel17}, it is seen how to rescale the number of particles in a finite bath correspond with to rescale the parameter $q$ associated with q-exponential distribution which governs the system as 
Interestingly, the rescaling  
\begin{equation}\label{escale}
q_\alpha=1+\alpha(q-1),
\end{equation}
for the parameter $q$, so much relevant in deformed algebra \cite{Borges04} and Tsallis thermostatistics \cite{Korbel17}, naturally appears in the entropic moment $W_q$ of the differential-escort distributions as shown in the next property.
\begin{property}\textbf{Rescaling of the entropic moments}
\label{LP_prop}\\        
Let $\rho$ be a probability distribution, $a \in\mathbb{R}$ and $\alpha \in\mathbb{R}$. Then the entropic moments $W_q[\rho]$ transform as 
\begin{equation}
W_q[\rho_\alpha]=W_{q_\alpha}[\rho],
\end{equation}
where $q_\alpha$ is given in \eqref{escale}. 

%\textcolor{red}{Moreover, this is true even when the entropic moments are divergent, in the following sense $\frac{W_q[\rho_\alpha]}{W_{q_\alpha}[\rho]}=1$
%% .}

\begin{proof}
If $W_q[\rho_\alpha]<\infty$, then 
$$W_q[\rho_\alpha]=\int_{\mathbb R}[\rho_\alpha(y)]^q\,dy=\int_{\mathbb R}[\rho(x)]^{\alpha q}\,[\rho(x)]^{1-\alpha}\,dx=W_{1+(q-1)\alpha}[\rho].$$
In case that $W_q[\rho_\alpha]=\infty$, we consider the following equality between finite integrals
$$\int_{y_1}^{y_2}[\rho_\alpha(y)]^q\,dy=\int_{x_1}^{x_2}[\rho(x)]^{q_\alpha}\,dx$$
for any $x_1,x_2\in\Lambda$ and $y_{1,2}=y(x_{1,2})$. So, one has
$$\frac{W_q[\rho_\alpha]}{W_{q_\alpha}[\rho]}=\lim_{(x_1,x_2)\to (x_{m},x_{M})}\frac{\int_{y_1}^{y_2}[\rho_\alpha(y)]^q\,dy}{\int_{x_1}^{x_2}[\rho(x)]^{q_\alpha}\,dx}=\lim_{(x_1,x_2)\to (x_{m},x_{M})} 1=1.$$
\end{proof}

\end{property}
For completeness, note that when  $q=1$, then $q_\alpha=1$  and both $W_1[\rho_\alpha]=W_1[\rho]=1$ as one expects.\\
%\textcolor{green}{(¿HABLAR DE IMPLICACIONES EN EL HAUSDORFF ENTROPIC MOMENT PROBLEM?).\\ It is worth  noting that this rescaling behavior implies that, given a numerable set of positive numbers $w_1,w_2,\cdots$ satisfying the conditions of the  Hausdorff entropic moment problem \cite{Romera01}, not only unambiguously characterize just a probability density $\rho(x),$ for which $w_n=W_n[\rho]$, but also the full family of deformed densities $\rho_\alpha$, for which $w_n=W_{1+\frac{n-1}\alpha}[\rho_\alpha]$. In the case $\alpha=0$, the density $\rho_\alpha$ becomes to the uniformity and exactly equal to one in all its support, and just we lost all information of the corresponding entropic moments, except for the norm $w_1=W_1[\rho_\alpha]=1$.}

The rescaling behavior in this property is automatically inherited by the differential Shannon, Rényi and Tsallis  entropies, as pointed out in the next property.
\begin{property}\label{entr-prop}
\textbf{Entropies transformations}\\
Let $q,\alpha\in\mathbb R$ and $\rho$ be a probability distribution. Then, the differential  Shannon, R\'enyi and Tsallis entropies of the differential-escort distributions transform as

\begin{equation}
\frac{S[\rho_\alpha]}{S[\rho]}=\frac{R_q[\rho_\alpha]}{R_{q_\alpha}[\rho]}=\frac{T_q[\rho_\alpha]}{T_{q_\alpha}[\rho]}=\alpha.
\end{equation}

%\textcolor{red}{even when the entropies become to be infinite!. }
\begin{proof}
Taking into account that $\frac{1-q_\alpha}{1-q}=\alpha$, the equality for the Rényi and Tsallis entropies trivially follows from property \ref{LP_prop}. \\
The Shannon case could be simply understood as the limit case $q\to1$, however, for the sake of illustration, we give the pretty simple and nice natural proof: 
\begin{equation*}
S[\rho_\alpha]=-\int_{\Lambda_\alpha}\rho_\alpha(y)\log[\rho_\alpha(y)]\,dy=-\int_{\Lambda}\rho(x)\log[\rho(x)^\alpha]\,dx=\alpha\,S[\rho].
\end{equation*}
%\end{equation*}
%
% $\frac{R_q[\rho_\alpha]}{R_{q_\alpha}[\rho]}=\frac{T_q[\rho_\alpha]}{T_{q_\alpha}[\rho]}=\alpha$
%regardless of the convergence of the integrals. 
%For any probability distribution $\rho$ and any pair of points $x_1,x_2\in\Lambda$, and denoting  $y_{1,2}=y(x_{1,2})$ when $y(x)$ is the biyection defined in (\ref{def-y}), one has that the folllowing identity   
%\begin{equation*}
%\int_{y_1}^{y_2}\rho_\alpha(y)\log[\rho_\alpha(y)]\,dy=\int_{x_1}^{x_2}\rho(x)\log[\rho(x)^\alpha]\,dx
%\end{equation*}
%holds; that is to say 
%$$\frac{\int_{y_1}^{y_2}\rho_\alpha(y)\log[\rho_\alpha(y)]\,dy}{\int_{x_1}^{x_2}\rho(x)\log[\rho(x)]\,dx}=\alpha.$$
%So, taking $x_1,x_2$ as the extreme points of the support of $\rho$ one obtains that 
%$\frac{S[\rho_\alpha]}{S[\rho]}=\alpha.$
%On the other hand, for the Rényi and Tsallis entropies, taking into account that property \ref{LP_prop} and that $\frac{\alpha}{1-q_\alpha}=\frac1{1-q}$ it trivially follows 
%
% $\frac{R_q[\rho_\alpha]}{R_{q_\alpha}[\rho]}=\frac{T_q[\rho_\alpha]}{T_{q_\alpha}[\rho]}=\alpha$
%regardless of the convergence of the integrals. 
\end{proof}
%
%\begin{eqnarray}\label{eq:Tsallis_alpha}\nonumber
%T_\lambda[\rho_\alpha]&=&\alpha\, T_{1+(\lambda-1)\alpha}[\rho] .
%\end{eqnarray}

\end{property}
 
Finally, as a direct consequence of the Jensen inequality we can assert that the Rényi entropy $R_q$  of the transformed density $\rho_\alpha$ is a concave function of $\alpha$ when $q>1$ and convex when $q<1$. Just as property (\ref{entr-prop}) claims, Shannon entropy has a linear behavior with the deformation parameter  $\alpha$. This behavior is given in the following property.
 
\begin{property}
\label{convexity_prop}

Let $\rho$ be a non-uniform probability density, then the R\'enyi differential entropy of the differential-escort distribution fulfills the following identity:
\begin{equation}
sgn\left(\frac{\partial^2R_q[\rho_\alpha]}{\partial\alpha^2}\right)=sgn(1-q).
\end{equation}
So, $ R_q[\rho_\alpha]$ is concave with $\alpha$ for $q>1$ and convex for $q<1$. When $\rho$ is an uniform density, then one  has $\frac{\partial^2R_q[\rho_\alpha]}{\partial\alpha^2}=0$.

\begin{proof}

One can easily compute
$$\frac{\partial^2 R_q [\rho_\alpha]}{\partial \alpha^2}=\frac{1-q}{\left(\int_{\Lambda}\rho^{q_\alpha}\right)^2} \left[\int_{\Lambda} \rho^{q_\alpha}\log^2\rho\int_\Lambda \rho^{q_\alpha}-\left(\int_\Lambda\rho^{q_\alpha}\log\rho\right)^2\right],$$
for any probability density.  On the other hand, due to Jensen's inequality one has
$$\left(\frac{\int_\Lambda \rho^{q_\alpha}\log\rho}{\int_\Lambda\rho^{q_\alpha}}\right)^2\le \frac{\int_\Lambda \rho^{q_\alpha}\log^2\rho}{\int_\Lambda \rho^{q_\alpha}},$$
where the equality holds if and only if $\rho$ is an uniform probability density. 
So, for non-uniform probability densities, it is straightforward to have that $sgn\left(\frac{\partial^2 R_q [\rho_\alpha]}{\partial \alpha^2}\right)=sgn(1-q).$
\end{proof}
\end{property}

\section{The LMC-Rényi Monotonicity} \label{sec:LMC-R}
The concept of monotonicity of a complexity measure was recently presented in \cite{Rudnicki16} and proven for the Fisher-Shannon and Crámer-Rao complexity measures. In this section we analyse the behavior of the  LMC-Rényi complexity measure under the differential-escort transformation, and then we show its monotonicity property. Let us first recall that the LMC-Rényi complexity measure is defined \cite{Sanchez-Moreno14,Lopez-Ruiz03,Lopez-Ruiz12} as
   \begin{equation}\label{def:LMC}
      C_{p,q}[\rho]=e^{R_p[\rho]-R_q[\rho]},\quad p<q.
   \end{equation}
  Note that the case ($p\to1, q= 2)$ corresponds to the plain LMC complexity measure \cite{Catalan02} 
      \begin{equation}
      C_{1,2} [\rho] = D[\rho] e^{S[\rho]} ,
      \end{equation}
which quantifies the combined balance of the average height of $\rho(x)$ (also called disequilibrium $D[\rho] = e^{R_2[\rho]}$), and its total spreading. This measure has been related with the \textit{degree of multifractality} of the distribution \cite{JMAngulo14} and widely applied in various contexts from electronic systems to seismic events \cite{Lopez-Ruiz03,Sen11,Pennini17}.
It satisfies interesting mathematical properties, such as invariance under scaling and translation transformations, invariance under replication and has a lower bound \cite{Lopez-Ruiz03} which is achieved by the uniform densities.
Obviously, this complexity measure inherits the regularity of the previous section which together with property \ref{entr-prop} allows us to write
\begin{property}\label{property7}
Let $p<q$ and $\alpha \in\mathbb{R}$. Then, the LMC-R\'enyi complexity of the probability distribution $\rho$  transforms as
\begin{equation}
C_{p,q}[\rho_\alpha]=\left(C_{p_\alpha,q_\alpha}[\rho]\right)^\alpha.
\end{equation}
\end{property}

Moreover a straightforward application of the Jensen inequality allows one to find 
\begin{property}\label{minimalbound}
Let $p<q$. Then, the LMC-R\'enyi complexity of the probability distribution $\rho$ is bounded as
 \begin{equation}
 C_{p,q}[\rho]\ge 1,
 \end{equation}
and the equality trivially holds when $\rho$ belongs to the class $\Xi$ of uniform distributions:
\begin{equation}
\Xi=\{\chi^{(a)}(x-x_0)|\, a>0,\, x_0\in\mathbb R\},\quad \chi^{(a)}(x)=\left\{\begin{array}{c} a^{-1}, x\in[0,a]\\0,\text{otherwise}
\end{array}\right..
\end{equation}
\end{property}

So, the LMC-Rényi complexity measure is universally bounded, and the family of minimizing densities is given by the class of uniform densities $\Xi$. Actually, this class remains invariant under differential-escort transformations. In fact, restricting us to $\Xi$, the transformation $\mathfrak E_\alpha$ just corresponds with a scaling change.
%; on the contrary, for the standard escort case any uniform distribution $\chi_a$ remain unchanged. Then, one has
\begin{property}\label{uniforms}
\textbf{Uniformity transformations}\\
 Let $\alpha \in\mathbb{R}.$ Then, 
\begin{equation}
\rho\in\Xi\Longleftrightarrow\mathfrak E_\alpha [\rho]\in\Xi,\quad\alpha\not=0. 
\end{equation}
Particularly, one has that $\mathfrak E_\alpha[\chi^{(a)}]=\chi^{(a^\alpha)}.$ 
%\footnote{Note that the uniform and unitary density $\chi^ {(1)}$ plays an special role since for any real $\alpha$, one has $\mathfrak E_\alpha[\chi^{(1)}]=\chi^{(1)}$.} 
%Moreover for any probability distribution, not necessarily uniform, one has that $\mathfrak E_0 [\rho]=\chi^{(1)}$.
\begin{proof}
 	For any real $\alpha$, one has
 	$$[\chi^{(a)}(x)]^\alpha=a^{-\alpha},\quad \forall x\in[0,a];$$ 
 	and $dy=a^{\alpha-1}dx$ from Eqs. (\ref{def-E}) and (\ref{def-y}), respectively. Then, with $y(0)=0$ one has that $y(x)$ obeys the linear relation $y(x)=a^{\alpha-1}x$. And by taking into account that $y([0,a])=[0,a^\alpha]$ one obtains
 	\begin{equation*}
 	\mathfrak E_\alpha[\chi^{(a)}](y)=[\chi^{(a)}(x(y))]^\alpha=a^{-\alpha},\quad \forall y\in[0,a^\alpha];
 	\end{equation*}
 	or equivalently $\mathfrak E_\alpha[\chi^{(a)}]=\chi^{(a^\alpha)}$.  	
\end{proof}
\end{property} 
Let us now show that the LMC-R\'enyi complexity measure is monotone with respect to the class of differential-escort transformations $\{\mathfrak E_\alpha\}_{\alpha\in[0,1]}$ in the Rudnicki et al sense; this means that $C[\mathfrak E_\alpha[\rho]] \le C[\rho] $  for any density $\rho$. We will see that this inequality is a direct consequence of the concavity of the Rényi entropy \ref{convexity_prop} with respect to the parameter of the deformation $\alpha$.\\

 First we observe that
 \begin{equation}\label{second-derivates}
 \frac{\partial^2 R_q[\rho_\alpha]}{\partial q\partial \alpha }=\frac{-\alpha}{1-q}\frac{\partial^2 R_q[\rho_\alpha]}{\partial \alpha^2 },
 \end{equation}    
which together with property \ref{convexity_prop} gives
\begin{equation} 
sgn\left(\frac{\partial^2 R_q[\rho_\alpha]}{\partial q\partial \alpha }\right)=-sgn(\alpha),\quad \rho\notin\Xi.
\end{equation}
Then, if we consider the derivative with respect to $\alpha$ we have that
$
           \frac{\partial C_{p,q}[\rho_\alpha]}{\partial \alpha}=C_{p,q}[\rho_\alpha]\left(\frac{\partial R_{p}[\rho_\alpha]}{\partial \alpha}-\ \frac{\partial R_{q}[\rho_\alpha]}{\partial \alpha}\right),$ and so taking into account that $p<q$ one has 
\begin{equation}\label{eq:sgn}
sgn\left(\frac{\partial C_{p,q}[\rho_\alpha]}{\partial \alpha}\right)=-sgn\left(\frac{\partial^2 R_{q}[\rho_\alpha]}{\partial q\partial \alpha}\right)=sgn(\alpha),\quad \rho\notin\Xi.
\end{equation}

So, from Eq. \eqref{eq:sgn} it trivially follows the searched property:
\begin{property}\label{prop:monot}
Let $p<q$ . Then, the LMC-R\'enyi complexity of the probability distribution $\rho$ fulfills that
  \[
      C_{p,q}[\rho_{\alpha'}]\ge C_{p,q}[\rho_{\alpha}],
      \]  
  for any $\alpha'>\alpha>0$ or $\alpha'<\alpha<0$. 
%  Particularly, for $\alpha'=1$ and $\alpha\in[0,1]$ one has
%    \[
%    C_{\lambda,\beta}[\mathfrak {E}_\alpha[\rho]]\le C_{\lambda,\beta}[\rho].
%    \]
%that is to say, the condition \ref{(iii)} is fulfilled with $\mathcal O=\{\mathfrak E_\alpha\}_{\alpha\in[0,1]}. $ 
Moreover, if $\rho\notin \Xi $ the equality only holds for $\alpha=1$ and the minimal value is only obtained  when $\alpha=0$. In the case that $\rho\in \Xi$, then     $C_{p,q}[\mathfrak {E}_\alpha[\rho]]=C_{p,q}[\rho]=1.$ 

Even more, for $\alpha=0$ the minimal possible value of the complexity measure is reached, $C_{p,q}[\mathfrak E_0 [\rho]]=1$. That is due to, for any $\rho$ one has that $\mathfrak E_0 [\rho]=\chi^{(1)}$. 

\end{property}
%\textcolor{blue}{\textcolor{red}{(*)}Note that, taking into account that $\rho_{\alpha=1}=\rho$ in property \ref{prop:monot} we can reduce or increase the LMC-Rényi complexity of the distribution $\rho$ until the extreme cases $\alpha=0$ or $\alpha\to+\infty$ respectively.}

The last three properties can be summarized by means of the following theorem: 
\begin{theorem}\label{theorem}
Given the family of uniform distributions $\Xi$, and the class of transformations $\mathfrak E_\alpha$, then  the triplet $\left( C_{p,q},\Xi,\mathfrak E_\alpha\right)$ satisfies the monotonicity property of the LMC-R\'enyi measure of complexity.
\end{theorem}            

The comparison of this result with the monotonocity property of the Crámer-Rao and Fisher-Shannon complexity measures obtained by Rudnicki et al. \cite{Rudnicki16} allows us to observe that the class of differential-escort operations plays for the LMC-R\'enyi measure of complexity the same role than the class of convolution-with-the-Gaussian operations in the Crámer-Rao and Fisher-Shannon cases.

%In this manner, at least in the sense of the monotonicity of the complexity measures defined by Rudnicki et al. \cite{Rudnicki16}, we can say that the class of \textit{differential escort} and the  class of convolution with the Gaussian play the same role with respect to the  LMC-Rényi complexity and the Crámer-Rao and Fisher-Shannon ones respectively. 

%Within each equivalence class defined by the differential-escort transformation, the probability densities have a relation of total order with respect to their LMC-Rényi complexity.  
%In this sense, each equivalence class represents a set of probability densities which are essentially similar, with the following  exception: they are differently far of the uniformity, where this "distance" with the uniformity is measured by the LMC-Rényi complexity, and the probability densities are naturally related through the operation $\mathfrak E_\alpha$. 

%In the following section we will study one the most important family of densities, namely, the Tsallis q-exponential. Actually we will show that this family is stable with respect the \textit{differential escort} densities. 

\section{Low and high complexity limits}\label{sec:extreme}
In this section we conduct a study of the behavior of the statistical properties of a general density, when deformed in extreme cases $\alpha\sim0$ and $\alpha\to +\infty $. To this end, we will first give three statements for the general case that will be useful in the study of the limit cases.

\begin{proposition}\label{prop1}
Let $\rho(x)$ a bounded density, then the entropic moments $W_q[\rho]$ satisfies
$$ W_q[\rho]<\infty \Longleftrightarrow q>q_c[\rho],$$
with $q_c[\rho]<1$.
\begin{proof}
Given a bounded probability density $\rho(x)$ then the proof trivially follows taking into account that Rényi entropy $R_q$ is decreasing in $q$, and that $R_q[\rho]\ge R_\infty[\rho]=-\log(\rho_{max})$.%On the other hand, $W_q[\rho]=e^{(1-q)R_q[\rho]}$, so for $q>1$ one has that $W_q[\rho]<\infty \Longleftrightarrow e^{(q-1)R_q[\rho]}\neq 0$, but $e^{R_q[\rho]}\ge\frac1{\sup_x\rho(x)}$, and so when $\rho(x)$ is upper bounded necessarily $W_q[\rho]<\infty,\,\quad \forall q>1$. On the other hand, when $q<1$ one has that $W_q[\rho]<\infty \Longleftrightarrow R_q[\rho]<\infty$, and taking into account that Rényi entropy is decreasing in $q$ one finds the previous proposition where $q_c[\rho]=\min\{q/\,\,W_q[\rho]=\infty\}$. Finally, taking into account that $W_1[\rho]=1$ and $W_{q\ge1}[\rho]<\infty$, then $q_c[\rho]<1$. 
\end{proof} 
\end{proposition}

For example, for an exponential-like decaying density one has $q_c[\rho]=0$, but for a power-law decaying density as $\mathcal O\left(x^{-\beta}\right)$ then $q_c[\rho]=1/\beta\in(0,1)$, or for any N-piecewise density $q_c[\rho]=-\infty$. 
On the other hand, it is easy to see that
\begin{equation}\label{eq:lambdac}
q_c[\rho_\alpha]=1-\frac{1-q_c[\rho]}{\alpha}.
\end{equation}
Deserves noting that if we take $\alpha_c=1-q_c[\rho]$ then $q_c[\rho_{\alpha_c}]=0$, what means that $\rho_{\alpha_c}$ has an infinite support $W_0[\rho_{\alpha_c}]=W_{q_c}[\rho]=\infty$, but all entropic moments with positive parameter $q$ are finite.

On the other hand, is easy to see that the LMC-Rényi complexity measure is not only bounded inferiorly but also superiorly.  
\begin{proposition}\label{prop2}
For any density $\rho\notin\Xi$ and any pair $p<q$ then
\begin{equation}\label{ineq:C}
1< C_{p,q}[\rho]<C_{p,\infty}[\rho]=\frac{\rho_{max}}{ \left\langle \rho^{p-1}\right\rangle^{\frac1{p-1}}}\,,
\end{equation}
contrary if $\rho\in\Xi$ then $C_{p,q}[\rho]=C_{p,\infty}[\rho]=1.$
%\textcolor{red}{\begin{proof}
%The proof trivially follows taking into account that differential Rényi entropy $R_q$ is decreasing in $q$ and that $R_{\infty}=-\log\sup_x\rho(x)$
%\end{proof}}
\begin{proof}
		Note that, if $\rho\notin\Xi$ and  $q'>q>p$, then from Eq. \eqref{def:LMC} and property \ref{minimalbound} it follows that $C_{p,q'}[\rho]=C_{p,q}[\rho]C_{q,q'}[\rho]>C_{p,q}[\rho]$. So taking $q'\to\infty$ one obtains Eq. \eqref{ineq:C}. On the other hand, as is claimed in property \ref{minimalbound}, if $\rho\in\Xi$ then $C_{p,q}[\rho]=1, \forall p<q$.
\end{proof}
\end{proposition}

Let us now introduce the notion of \textit{entropic cumulant} of order $n$ of a probability density $\rho$. Note that these entropic cumulants have the same structure than the ordinary cumulants $k_n=\left.\frac{d^n\log\langle e^{px}\rangle}{dp^n}\right|_{p=0}$.
\begin{definition}\label{def:cum}
	Let $n\in\mathbb N$ and $\rho\in\mathcal D(\mathbb R)$, so  the entropic cumulant of order $n$, $\mathfrak K_n[\rho]$, is defined as 
\begin{equation}
	\mathfrak K_n[\rho]=\left.\frac{d^n \log \langle \rho^{q-1}\rangle}{dq^n}\right|_{q=1}.
	\end{equation}
	  Particularly, 
\begin{eqnarray*} 
	\mathfrak K_0[\rho]&=& 0, \\ \mathfrak K_1[\rho]&=&-S[\rho]=\langle\log\rho\rangle\\
	\mathfrak K_2[\rho]&=&\langle\log^2\rho\rangle-\langle\log\rho\rangle^2,\\ \mathfrak K_3[\rho]&=&\langle\log^3\rho\rangle-3\langle\log^2\rho\rangle\langle\log\rho\rangle+2\langle\log\rho\rangle^3 \\
	\cdots
\end{eqnarray*}
\end{definition}
It is worth to mention that, when $\rho$ is an uniform density, then $\mathfrak K_n[\rho]=0,\quad \forall n>1$; this behavior is similar to the Gaussian probability density with respect to the ordinary cumulants.\\
This definition and the next property will be useful in the following propositions \ref{prop3}, \ref{prop6} and  \ref{prop9}.
\begin{property}\label{prop:cumulants}
	Given any probability density $\rho$ and any $\alpha\in\mathbb R$ then 
	\begin{equation}
	\mathfrak K_n[\rho_\alpha]=\alpha^n\,\mathfrak K_n[\rho],\quad n\in\mathbb N
	\end{equation}
\end{property}
and then it follows that $\frac{\mathfrak K_{n+1}[\rho_\alpha]}{\mathfrak K_{n}[\rho_\alpha]}=\alpha\frac{\mathfrak K_{n+1}[\rho]}{\mathfrak K_{n}[\rho]}$ . 
	\begin{proof}
		From definition \ref{def:cum} one has that
		$$\mathfrak K_{n}[\rho_\alpha]=\left.\frac{d^n}{dq^n} \log W_q[\rho_\alpha]\right|_{q=1}=\left.\frac{d^n}{dq^n} \log W_{q_\alpha}[\rho]\right|_{q=1}=\alpha^n\left.\frac{d^n}{dq^n} \log W_{q}[\rho]\right|_{q=1}=\alpha^n \mathfrak K_{n}[\rho],$$
		where we have used property \ref{LP_prop} and taken into account that, for all $\alpha\neq0$ then $q=1\Longleftrightarrow q_\alpha=1$.
	\end{proof}

Finally, the third proposition is achieved through the Taylor series of the Rényi entropy $R_q[\rho]$ on its entropic parameter around $q=1$.

\begin{proposition}\label{prop3}
Given any probability density $\rho$, then the associated LMC-Rényi measure can be formally expressed as
\begin{eqnarray*}
C_{p,q}[\rho]&=&e^{\frac{\mathfrak K_{2}[\rho]}{2}(q-p)}\,\prod_{n=2}^\infty e^{\,\frac{\mathfrak K_{n+1}[\rho]}{(n+1)!}\left[(q-1)^{n}-(p-1)^{n}\right]},
\end{eqnarray*}
%	\textcolor{blue}{where the quantities $\mathfrak K_n[\rho]=\left.\frac{d^n \log \langle \rho^{q-1}\rangle}{dq^n}\right|_{q=1}$  have the same structure than the cumulants $k_n=\left.\frac{d^n\log\langle e^{px}\rangle}{dp^n}\right|_{p=0}$, i.e., 
%\begin{eqnarray*}
%\mathfrak K_0[\rho]&=& 0, \\ \mathfrak K_1[\rho]&=&-S[\rho]=\langle\log\rho\rangle\\
%\mathfrak K_2[\rho]&=&\langle\log^2\rho\rangle-\langle\log\rho\rangle^2,\\ \mathfrak K_3[\rho]&=&\langle\log^3\rho\rangle-3\langle\log^2\rho\rangle\langle\log\rho\rangle+2\langle\log\rho\rangle^3 \\
%\cdots
%\end{eqnarray*}
%}
provided that the series is convergent. 

 \begin{proof}Let us consider the Taylor series of $\log W_q[\rho]$ around $q=1$, 
\begin{equation*}
R_{q}[\rho]=\frac{1}{1-q} \sum_{n=0}^\infty \left.\frac{d^n\,\left(\log\left\langle\rho^{q-1}\right \rangle\right)}{dq^n}\right|_{q=1} \frac{(q-1)^n}{n!}=-\sum_{n=0}^\infty \mathfrak K_n[\rho] \frac{(q-1)^{n-1}}{n!},
\end{equation*}
provided that the series is convergent. 
So, taking into account that $\mathfrak K_0=0$ one can write $ R_q[\rho]=-\mathfrak K_1-\sum_{n=1}^\infty \mathfrak K_{n+1}[\rho] \frac{(q-1)^{n}}{(n+1)!}$. So the LMC-Rényi complexity measure can be expressed as 
\begin{eqnarray*}
C_{p,q}[\rho]=e^{R_{p}[\rho]-R_{q}[\rho]}=e^{\sum_{n=1}^\infty \frac{\mathfrak K_{n+1}[\rho] }{(n+1)!}[(q-1)^{n}-(p-1)^{n}]}.
\end{eqnarray*}
\end{proof}
%Rényi entropy can be written as 
%\begin{eqnarray*}
%R_{\lambda}[\rho]&=&\frac1{1-\lambda}\sum_{n=0}^\infty \left.\frac{d^n \log \langle \rho^{\lambda-1}\rangle}{d\lambda^n}\right|_{\lambda=1}\,\frac{(\lambda-1)^n}{n!}=-\sum_{n=0}^\infty \left.\frac{d^n \log \langle \rho^{\lambda-1}\rangle}{d\lambda^n}\right|_{\lambda=1}\,\frac{(\lambda-1)^{n-1}}{n!}
%\\
%&= & S[\rho]-\sum_{n=1}^\infty \,\frac{\mathfrak K_{n+1}[\rho]}{(n+1)!}(\lambda-1)^{n}
%\end{eqnarray*}
%so 
%\begin{eqnarray*}
%R_{\lambda}[\rho]-R_{\beta}[\rho]&=&\frac{\mathfrak K_{2}[\rho]}{2}(\beta-\lambda)+\sum_{n=2}^\infty \,\frac{\mathfrak K_{n+1}[\rho]}{(n+1)!}\left[(\beta-1)^{n}-(\lambda-1)^{n}\right] 
%\end{eqnarray*}
Particularly, for the conventional LMC complexity measure one has that
\begin{equation*}
C_{1,2}[\rho]\equiv C_{LMC}[\rho]=e^{\frac{\mathfrak K_{2}[\rho]}{2}}\,e^{\frac{\mathfrak K_{3}[\rho]}{3!}}\,e^{\frac{\mathfrak K_{4}[\rho]}{4!}}\,\cdots
\end{equation*}
\end{proposition}
It is specially interesting that, for $p,q\sim1$ we can write
\begin{eqnarray}\label{eq:Co0}
C_{p,q}[\rho]&\simeq&e^{\frac{\mathfrak K_{2}[\rho]}{2}(q-p)}.
\end{eqnarray}
Moreover, when $\rho\in\Xi,$ then $\mathfrak K_2[\rho]=0$,  and so $C_{p,q}[\rho]=e^{\frac{\mathfrak K_{2}[\rho]}{2}(q-p)}=1,\quad \forall p<q.$

\subsection*{Low complexity}

 Given any probability density $\rho$ (with $C_{p,q}[\rho]<\infty$), and choosing a real number  $\alpha\simeq0$, then following the Theorem \ref{theorem} one can always consider that $\rho_\alpha$ is a \textit{low complexity density }(in the LMC-sense). First, we note that when $\alpha\to0$, Eq. \eqref{eq:lambdac} diverges, so 
 \begin{proposition}\label{prop4}
 Let $\rho(x)$ a bounded and \textit{low complexity} density, then the critical entropic parameter $q_c[\rho]<<0$. 
\begin{proof}
		Taking $\alpha\to0$ in Eq. \eqref{eq:lambdac} one obtains $q_c[\rho]\to-\infty$. 
\end{proof}
\end{proposition}
  
 On the other hand, the upper bound of the LMC-Rényi measure of a \textit{low complexity density} goes to the unity, in such way that Eq \eqref{ineq:C} is crushed. 

 \begin{equation}
  1< C_{p,q}[\rho_\alpha]<C_{p_\alpha,\infty}[\rho]^\alpha,\quad \alpha\simeq0.
  \end{equation}
So, it follows that
 \begin{proposition}\label{prop5}
For a \textit{low complexity} density $\rho$, one has that 
  \begin{equation}
   1< C_{p,q}[\rho]<C_{p,\infty}[\rho],
   \end{equation}
 but, $C_{p,\infty}[\rho]\simeq1.$
\begin{proof}
Given a positive $\alpha\simeq0$, and a probability density $\rho$ such that $C_{p_\alpha,\infty}[\rho]<\infty$, then due to property \eqref{property7} one has that  $C_{p,\infty}[\rho_\alpha]=(C_{p_\alpha,\infty}[\rho])^\alpha\simeq 1$
\end{proof}
  \end{proposition}
  
 Finally, taking the Taylor series of $R_q[\rho_\alpha]$ around $\alpha=0$ one obtains 
%  \begin{equation}
%R_\lambda[\rho_\alpha]\sim-\alpha S[\rho]+\frac{\alpha^2}{2}(1-\lambda)\left(\langle \log^2\rho\rangle-S[\rho]^2\right)
%  \end{equation}
% and taking into account that $\langle \log^2\rho_\alpha\rangle=\alpha^2\langle \log^2\rho\rangle$
%and $\langle \log\rho_\alpha\rangle=\alpha\langle \log\rho\rangle$ one can write 
$ C_{p,q}[\rho_\alpha]\sim e^{\alpha^2\mathfrak K_2[\rho]\,\frac{(q-p)}{2}}$ , but just taking into account property \ref{prop:cumulants}, then  $e^{\alpha^2\mathfrak K_2[\rho]\,\frac{(q-p)}{2}}=e^{\mathfrak K_2[\rho_\alpha]\,\frac{(q-p)}{2}}.$  That is to say,  we can assure that
\begin{proposition}\label{prop6}
If $\rho$ is a low complexity density, then for any fixed $p<q<<\infty$
\begin{equation}
C_{p,q}[\rho]\simeq e^{\frac{\mathfrak K_2[\rho]}2\,(q-p)}.
\end{equation}
\begin{proof}
Given a positive $\alpha\simeq0$, from proposition \ref{prop3} and property \ref{prop:cumulants} one has that
\begin{eqnarray*}
	C_{p,q}[\rho_\alpha]&=&e^{\frac{\alpha^2\mathfrak K_{2}[\rho]}{2}(q-p)}\,\prod_{n=2}^\infty e^{\alpha^{n+1}\,\frac{\mathfrak K_{n+1}[\rho]}{(n+1)!}\left[(q-1)^{n}-(p-1)^{n}\right]}\simeq e^{\frac{\alpha^2\mathfrak K_{2}[\rho]}{2}(q-p)}=e^{\frac{\mathfrak K_{2}[\rho_\alpha]}{2}(q-p)}.
\end{eqnarray*}
\end{proof}
\end{proposition}
In fact, note that taking into account property \ref{prop:cumulants}, the lowest \textit{entropic cumulants} $\mathfrak K_n[\rho]$ domain for the \textit{low complexity} densities. Moreover, in these cases one typically has that $|\mathfrak K_{n+1}[\rho]|<|\mathfrak K_n[\rho]|.$
%However, note that in the limit $\beta\to\infty$ this identity can not be fulfilled due to $C_{\lambda,\infty}[\rho]=N_\lambda[\rho]\,\rho_{max}$.
%In fact 
%\begin{equation}
%1\le C_{\lambda,\beta}[\rho_\alpha]\le N_\lambda[\rho_\alpha]\,(\rho_\alpha)_{max}
%\end{equation}
%\begin{equation}
%1\le C_{\lambda_\alpha,\beta_\alpha}[\rho]^\alpha\le ( N_{\lambda_\alpha}[\rho]\,\rho_{max})^\alpha
%\end{equation}
%
%Moreover, note any bounded density has finite entropic moments $W_\lambda[\rho]$ for any $\lambda\ge1$. On the other hand, this is not guarantee for values $\lambda<1$. Typically there is a critical value $\lambda_c$ such that $W_{\lambda}[\rho]<\infty$, for any $\lambda>\lambda_c.$ For example, for an exponential density $\lambda_c=0$, but for power-law decaying densities $\lambda_c\in (0,1)$. So, taking a probability density with a critical entropic moment of order $\lambda_c$, and a positive and small number $\alpha$, one has that $ W_{\lambda}[\rho_\alpha]=W_{\lambda_\alpha}[\rho]<\infty$ when $1+(\lambda-1)\alpha>\lambda_c$,i.e., $ \lambda>1-\frac{1-\lambda_c}{\alpha}$. 
%That is to say 
%\begin{equation}
%\lambda_c[\rho_\alpha]=1-\frac{1-\lambda_c}{\alpha}
%\end{equation}
%So, we can say that a \textit{low complexity} density has all the entropic moments well define.    
%
\subsection*{High complexity}
In order to explore the high complexity limit, one can take any non-uniform probability density $\rho$, and a very large $\alpha>>1$. So,  following Theorem \ref{theorem} one can claim that $\rho_\alpha$ is a \textit{high complexity} density.

 First of all, note that the critical entropic parameter $q_c[\rho]$ of a high complexity density is closed to one \eqref{eq:lambdac}, that is to say

\begin{proposition}\label{prop7}
 Let $\rho(x)$ a bounded and \textit{high complexity} density, then the critical entropic parameter $q_c[\rho]\lesssim1$. 

 \begin{proof}
 	Taking $\alpha>>1$ in Eq. \eqref{eq:lambdac} one obtains $q_c[\rho]\to1$. 
\end{proof}
  \end{proposition}

On the other hand, the inequality \eqref{ineq:C} losses the upper bound. So we can assure that
 \begin{proposition}\label{prop8}
For a \textit{high complexity} density $\rho$ one has that 
  \begin{equation}
   1< C_{p,q}[\rho]<C_{p,\infty}[\rho],
   \end{equation}
 but, $C_{p,\infty}[\rho]>>1,$ for any fixed $p<<\infty$. 

 \begin{proof}
Given a non uniform probability density $\rho$, with $C_{p,q}[\rho]>1$, then  one obtains $C_{p,\infty}[\rho_\alpha]=(C_{p_\alpha,\infty}[\rho])^\alpha\to\infty$ when $\alpha\to\infty$. 	
\end{proof}
   \end{proposition}
Finally, it deserves to note that, although Eq. \ref{eq:Co0} must be valid for values of the parameters $p$ and $q$ enough close to one, for fixed  $p$ and $q$ is possible to find a density enough complex, in such a way that Eq. \ref{eq:Co0} is not satisfied. In fact $C_{p,q}[\rho_\alpha]=C_{p_\alpha,q_\alpha}[\rho]^\alpha\simeq e^{\frac{\mathfrak K_2[\rho]}{2}(q-p)\,\alpha^2}=e^{\frac{\mathfrak K_2[\rho_\alpha]}{2}(q-p)}$, whenever $p_\alpha\simeq1$ and $q_\alpha\simeq1$; that is to say $\alpha(p-1)\simeq0$ and $\alpha(q-1)\simeq0$. 

Moreover, taking into account Eq. \eqref{prop:cumulants}, for a \textit{high complexity }density the highest order entropic cumulants $\mathfrak K_n[\rho]$ will be dominants. 
%Concretely
%\begin{equation}
%\frac{\mathfrak K_{n+1}[\rho_\alpha]}{\mathfrak K_{n}[\rho_\alpha]}=\alpha\frac{\mathfrak K_{n+1}[\rho]}{\mathfrak K_{n}[\rho]}
%\end{equation}
All these considerations are summarized in the next proposition.
\begin{proposition}\label{prop9}
If $\rho$ is a \textit{high complexity }density, then the domain of parameters $p,q$ for what Eq. \eqref{eq:Co0} remains valid is extremely tiny. In fact, the highest order entropic cumulants $\mathfrak K_n[\rho]$ domain the behavior of the LMC-Rényi complexity measure. Actually, one has that $|\mathfrak K_{n+1}[\rho]|>|\mathfrak K_n[\rho]|$. 

\end{proposition}

  \subsection*{Example}
  In the following we give an example with numerical values. Note that, due to LMC-Rényi is invariant under replication transformation the number $N$ of different regions does not play a relevant role in the behavior of this complexity measure. So, for our purpose it is enough a simple example with $N=3$. 
  
  We are going to represent an initial distribution with three steps whose heights are $h_1=\frac32, h_2=1, h_3=\frac12$ and their weights are $w_1=w_2=w_3=\frac13$. In Figure \ref{figesc} we show the complexity reduction process through the here studied transformation
  \begin{figure}[H]\centering
   \begin{minipage}{0.4\linewidth}
   \includegraphics[width=\linewidth]{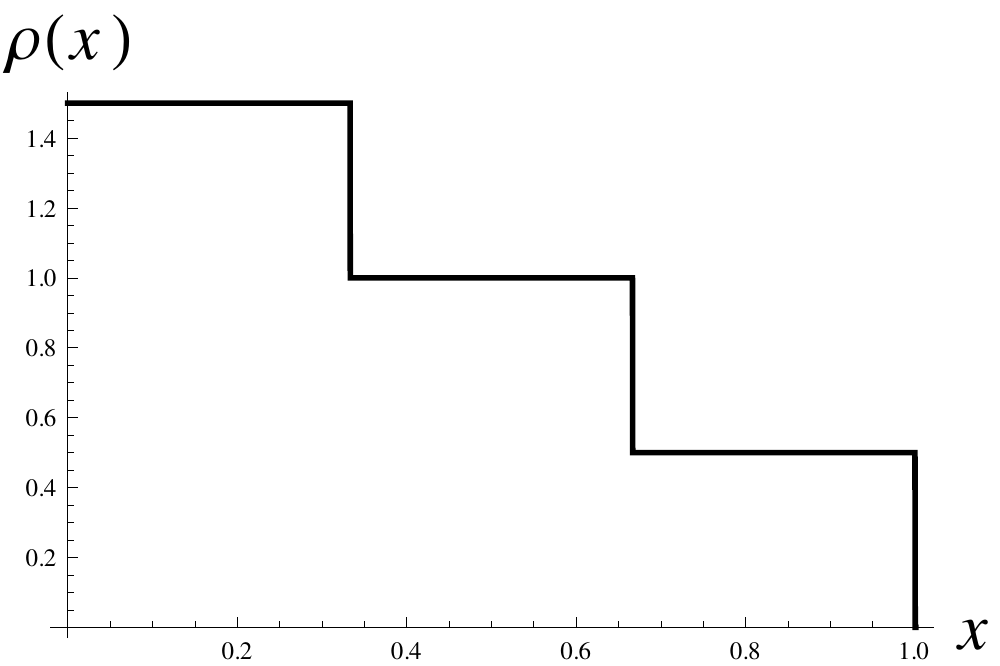}
   \end{minipage}
  \begin{minipage}{0.4\linewidth}
   \includegraphics[width=\linewidth]{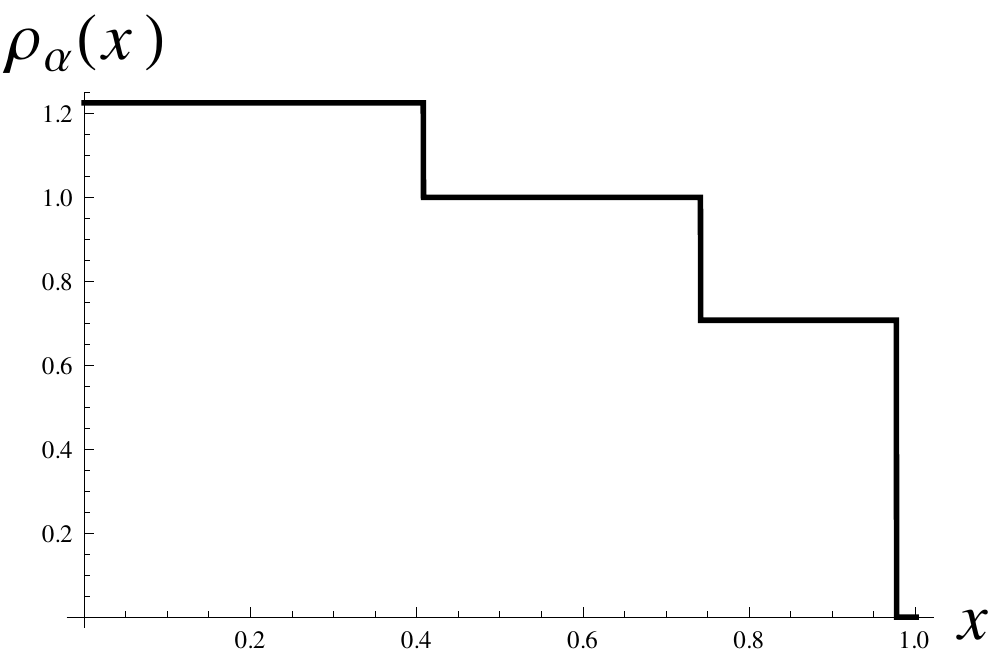}
     \end{minipage}
     \begin{minipage}{0.4\linewidth}
      \includegraphics[width=\linewidth]{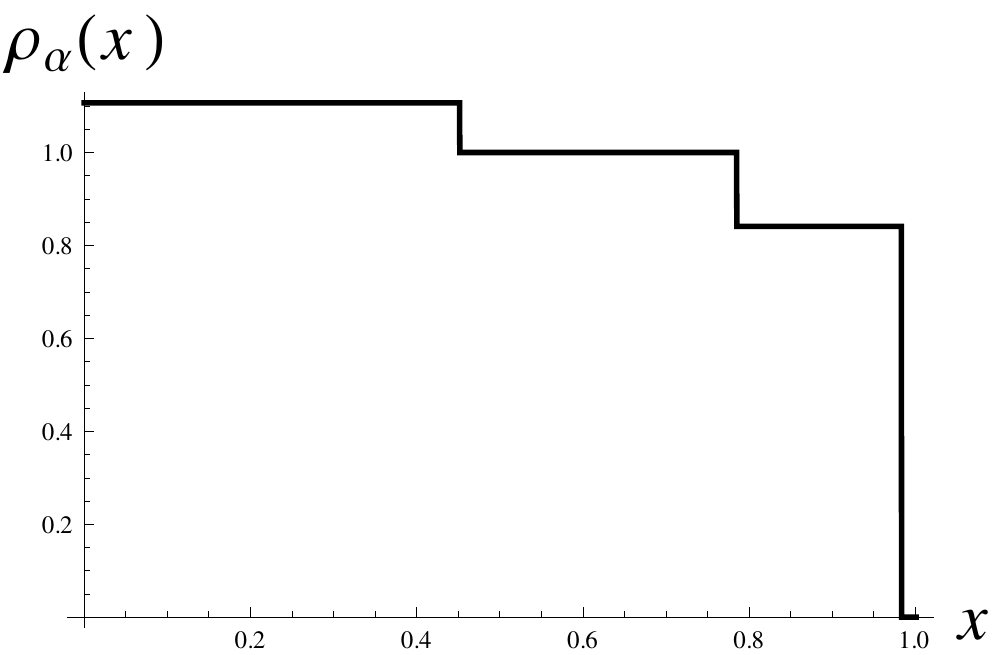}
  \end{minipage}
  \begin{minipage}{0.4\linewidth}
   \includegraphics[width=\linewidth]{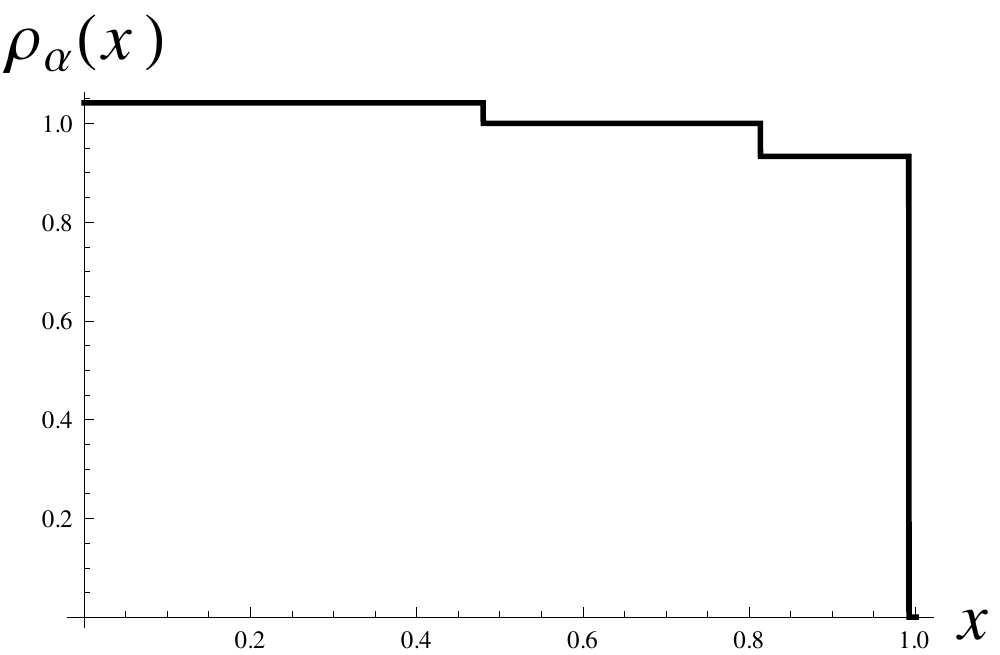}
   \end{minipage}
          \caption{Transformed density $\rho_\alpha (x)$ for different values of the transformation parameter $\alpha= 1, \frac12,\frac14,\frac1{10}.$}
          \label{figesc}
  \end{figure}
  It is interesting to give the values of the LMC complexity for these distributions, $C_{LMC}[\rho_\alpha]\simeq1.06923,1.01818,1.00468,1.00076$ for $\alpha=1,0.5,0.25,0.1$ respectively.
  In Figure \ref{figesc2} we represent the complexity increasing process of this probability density
  \begin{figure}[H]\centering
   \begin{minipage}{0.4\linewidth}
   \includegraphics[width=\linewidth]{./figures/piecewise.pdf}
   \end{minipage}
  \begin{minipage}{0.4\linewidth}
   \includegraphics[width=\linewidth]{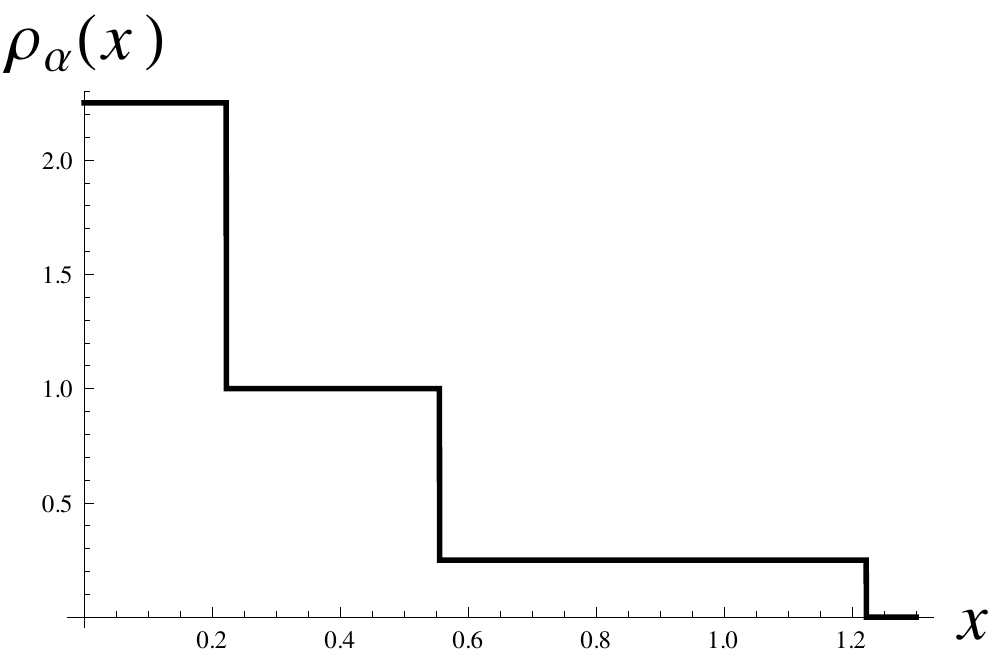}
     \end{minipage}
     \begin{minipage}{0.4\linewidth}
      \includegraphics[width=\linewidth]{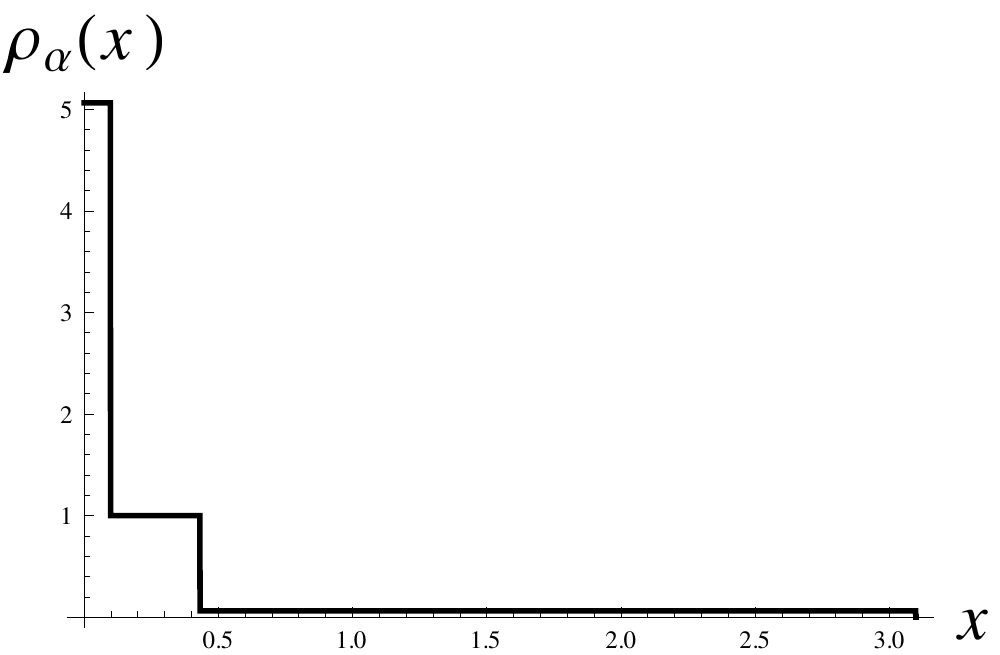}
  \end{minipage}
  \begin{minipage}{0.4\linewidth}
   \includegraphics[width=\linewidth]{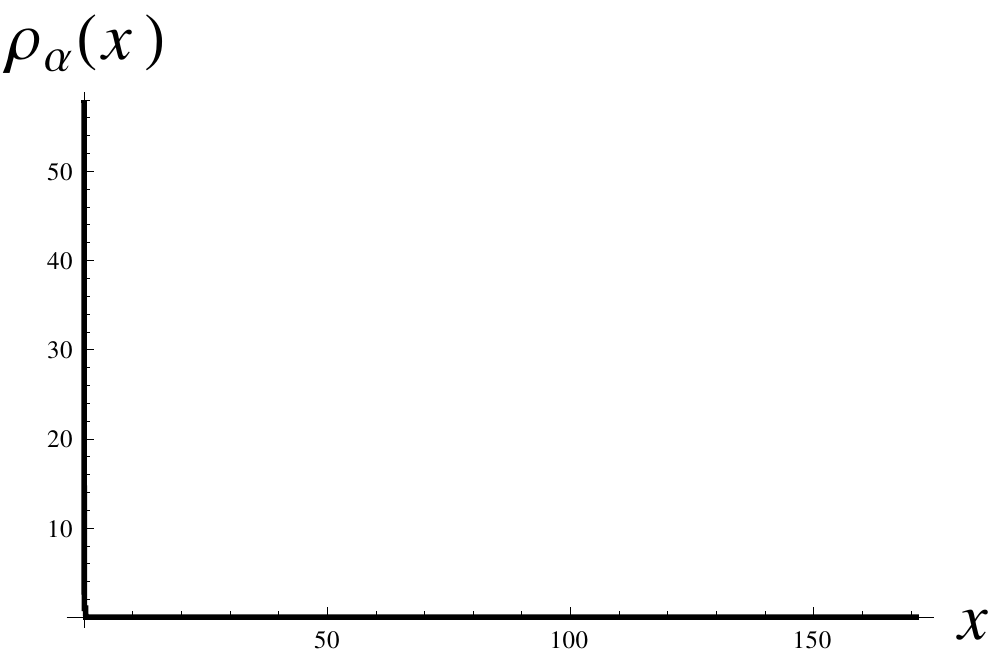}
   \end{minipage}
          \caption{Transformed density $\rho_\alpha (x)$ for different values of the transformation parameter $\alpha= 1, 2,4,10.$}
          \label{figesc2}
  \end{figure}
  Note that, in the case $\alpha=10$, one has that $(w_1)_\alpha\simeq 0.008$ and $(h_1)_\alpha\simeq57$, $(w_2)_\alpha\simeq 0.03$ and $(h_2)_\alpha=1$ and finally $(w_3)_\alpha\simeq 170$ and $(h_3)_\alpha\simeq0.001$. So, in this case the graphic representation is really difficult to be performance. For the sake of illustration, we give the case $\alpha=100$, for which  $(w_1)_\alpha\simeq 10^{-18}$ and $(h_1)_\alpha\simeq 4\times10^{17}$, $(w_2)_\alpha\simeq 0.03$ and $(h_2)_\alpha=1$ and finally $(w_3)_\alpha\simeq 2\times10^{29}$ and $(h_3)_\alpha\simeq7\times10^{-31}$, what seems to be near to impossible to be graphically represented with accuracy (even using a logarithmic scale in both axes) while still being a $3$-piecewise density.    The values of the LMC complexity for these densities are $C_{LMC}[\rho_\alpha]\simeq 1.06923,1.25988,2.02809, 12.1843,3\times10^{13}$ for $\alpha=1,2,4,10,100$ respectively. 
  
%  \textcolor{green}{
%   In the following Figure \ref{Fig1}, the most relevant properties of the \textit{differential escort} distributions are clearly shown, namely the concavity and convexity of the Rényi entropy with $\alpha$ depending on $sgn(\lambda-1)$, the linearity of the Shannon entropy, and the monotonicity of the LMC-Rényi complexity measure. In fact, Rényi entropy and LMC-Rényi complexity measure of a standard N-piecewise distribution are represented as function of the deformation parameter $\alpha$
%   \begin{figure}[H]\centering
%    \begin{minipage}{0.4\linewidth}
%    \includegraphics[width=\linewidth]{./figures/Renyis1.pdf}
%    \end{minipage}
%   \begin{minipage}{0.4\linewidth}
%    \includegraphics[width=\linewidth]{./figures/LMC-Renyi.pdf}
%      \end{minipage}
%            \caption{...}
%\label{Fig1}
%   \end{figure}
%   Note that the LMC-Rényi complexity is bounded when $1<\lambda<\beta$, or $\lambda<\beta<1$, what is easy to see analytically. In contrast when  $\lambda\le1\le\beta$ (of course with $\lambda<\beta$) the complexity can be unlimitedly increased.
% \\(\textbf{REVISAR O REHACER ESTUDIO NUMÉRICO})
%}
\section{q-exponential and power-law decaying densities} \label{sec:q-exp}

The exponential and q-exponential distributions are fundamental tools in the extensive and non-extensive formalisms \cite{Tsallis09}. They can be obtained by maximizing the differential  Rényi and Tsallis entropies with a suitable constraint\cite{Tsallis88}, or by maximizing differential Shannon  entropy with some tail constraints \cite{Bercher08}. In this section we will study the q-exponential distribution in the framework of the differential-escort transformations which will able to naturally relate it to the exponential one. 
      
      The exponential function $\mathcal E(x)= e^{-x}$, with $x\in[0,\infty)$, is recovered taking the limit $q\to1$ in the family of q-exponential functions defined by
      \begin{equation}
e_q(x)=(1+(1-q)x)_+^{\frac1{1-q}},
      \end{equation}
      where $(t)_+=\max\{t,0\}$.
Tsallis introduced \cite{Tsallis88} the q-exponential probability densities which are proportional to $e_q(-x)$.  For convenience, we denote the q-exponential densities as 
  \begin{equation}\label{eq:not-Eq}
\mathcal E_q (y)\equiv e_{q}\left(-\frac y{2-q}\right).
  \end{equation}
Note that when $q\in(1,2)$ the support is non compact and the tail of the probability density decays as a heavy-tailed distribution; in contrast when $q<1$, the support is compact.

It is worth to realize that the standard escort transformation of a q-exponential density is another q-exponential; indeed,
\begin{equation}
E_\alpha[\mathcal E_q]=\mathcal E_{q'},\quad q'=1+\frac{q-1}\alpha.
\end{equation}
Note that, if $q=1$ then $q'=1$; that is to say, the escort transformation of an exponential distribution is another exponential distribution. On the other hand, if $q>1$ the support of $\mathcal E_q$ is not compact and so necessarily $\alpha>q-1>0$ for the sake of satisfying the convergence condition given in (\ref{escort}); and in consequence, when $q\in(1,2)$ necessarily $q'\in(1,2)$. Finally, when $q<1$  one has that $q'<1$ for any $\alpha>0$. In other words, the escort transformation $E_\alpha$ keep unchanged the three regions of the parameter $q$ ($q<1, q=1, q>1$); this behavior is expected since  the standard escort transformation keep the support invariant. 

This behavior is totally different for the differential-escort transformation, which indeed changes the length of the support. In fact, it transforms not only a q-exponential distribution in another one, but also: given any initial value of the parameter $q<2$, any other parameter $q'<2$ can be obtained through the use of $\mathfrak E_\alpha$ with $\alpha\neq1$,  as we shall see  below. 
      
From definition \ref{definition}, given any $\alpha$ one has that 
\begin{equation}\label{eq:EE}
\mathfrak E_\alpha[\mathcal E](y)=e^{-\alpha\, x(y)},
\end{equation} 
with
      \begin{equation}\label{eq:exp-ch}
y(x)=\int_0^xe^{(\alpha-1)t}\,dt=\frac{1}{\alpha-1}\left(e^{(\alpha-1)x}-1\right),\quad \alpha\neq1;
      \end{equation}
and so one easily obtains

\begin{equation}\label{eq:x(y)}
x(y)=\frac1{\alpha-1}\log\left(1+(\alpha-1)y\right).
\end{equation}
So, inserting \eqref{eq:x(y)} in \eqref{eq:EE} we have that            
\begin{equation}\label{eq:exp1}
\mathfrak E_\alpha[\mathcal E](y)=\left(1+(\alpha-1)y\right)^{\frac{\alpha}{1-\alpha}},
\end{equation}
where, from Eq. \eqref{eq:exp-ch}, $y\in[0,\infty]$ for $\alpha>1$ and $y\in[0,\frac{1}{1-\alpha}]$ when $\alpha<1$. In fact,  we can rewrite Eq. \eqref{eq:exp1} as
\begin{equation}\label{eq:exp2}
 \mathfrak E_\alpha[\mathcal E](y)=e_{\frac{2\alpha-1}\alpha}(-\alpha y).
 \end{equation}
Or equivalently, choosing $\alpha=\frac1{2-q}$ and using the notation introduced in Eq. \eqref{eq:not-Eq}, one can write

 \begin{equation}\label{eq:exp3}
 \mathfrak E_{\frac1{2-q}}[\mathcal E]=\mathcal E_q. 
 \end{equation}         

On the other hand,  taking into account that $\mathfrak E_1[\rho]=\rho$ and considering the composition property \ref{compo-prop} and Eq. \eqref{eq:exp3} one obtains the identities
      
      \begin{equation}
     \mathfrak E_{2-q}[\mathcal E_q]=\mathfrak E_{2-q}[\mathfrak E_{\frac{1}{2-q}}[\mathcal E]]=\mathfrak E_{\frac{2-q}{2-q}}[\mathcal E]= \mathcal E,\quad \forall q<2.
      \end{equation}         
      
%    \textcolor{blue}{Note that a $(2-q)$ \textit{duality/inverse} property is satisfied in the q-deformed algebra \cite{Tsallis09}}.
     From which, taking any couple $q,\tilde q<2$ one can write the following relation between q-exponential densities
\begin{equation}
\mathfrak E_{{2-q}}[\mathcal E_q]=\mathfrak E_{{2-\tilde q}}[\mathcal E_{\tilde q}],
\end{equation}         
or equivalently, using again the composition property,
     \begin{equation}
\mathfrak E_{\frac{2-q}{2-\tilde q}}[\mathcal E_q]=\mathcal E_{\tilde q},
\end{equation}
or as well $\mathfrak E_{\alpha}[\mathcal E_q]=\mathcal E_{\overline q}$ with 

\begin{equation}
\overline q=2+\frac{q-2}\alpha.  
\end{equation}
Thus, as we have previously anticipated, starting with any $q<2$ we can obtain any other value $\overline q<2$. In particular, when $\overline q>1$ it occurs that $\alpha>2-q$, when $\overline q=1$ one has  $\alpha=2-q$, and  when $\overline q<1$ it happens that $\alpha<2-q$. Note that the value $\alpha=2-q$ plays a critical role. Finally, it is worth mentioning that, when $\alpha>0$ then $\overline q<2$, but taking $\alpha<0$ one obtains $\overline q>2$ which normally is not considered, however note that these densities are correctly defined and they satisfy the normalization condition $\int_{\Lambda_\alpha}\rho_\alpha(y)=1$.

These results are a little bit  extended in the next lemma: 
%
%Although the simplicity of the exponential case, typically the \textit{differential escort} distribution cannot be easily calculated for most distributions. In the next lemma we analyse the behavior for the tail of any probability distribution when it is deformed with the \textit{differential escort} transformation. 

 \begin{lemma}\label{lemma}
          Let $\rho(x)>0,\,\forall x\in\mathbb [0,\infty),$ be a probability density, such that the tail of $\rho(x)$ decreases as $\mathcal O(x^{-\beta}),\quad \beta>1$. Then, for $\alpha>\alpha_c=\frac{\beta-1}{\beta}$, the tail of the transformed distribution $\rho_\alpha(y)$ decreases as $\mathcal{O}\left(y^{\frac{-\beta\alpha}{1-\beta(1-\alpha)}}\right)$. On the other hand, for $\alpha<\alpha_c$, the distribution $\rho_\alpha$ has a compact support. Finally, when $\alpha=\alpha_c$ the support is non-compact and there is an exponential decay. 
\end{lemma}
      
\begin{proof}
              Let $\alpha\in\mathbb R$. 
              
              Given $x>>1$, the original density fulfilled $\rho(x)\sim x^{-\beta}$. On the other hand the variable change is defined as $y(x)=\int_0^x[\rho(t)]^{1-\alpha}\,dt$. 
               Then, the length of the support of $\rho_\alpha$ is given by $W_0[\rho_\alpha]=W_{1-\alpha}[\rho]=\int_0^\infty \rho(x)^{1-\alpha}dx\sim\int_{a>0}^\infty x^{-\beta(1-\alpha)}$. So, it is clear that the support of $\rho_\alpha$ is compact iff $\alpha<\frac{\beta-1}{\beta}$, and  in the case $\alpha\ge\frac{\beta-1}\beta$ we have that $\lim_{x\to\infty} y(x)\to\infty.$

In the case $\alpha\ge\frac{\beta-1}\beta$, one can suppose $x>>1$, and so $\rho_\alpha(y)\propto[x(y)]^{-\beta\alpha}$, and in the other hand 
               $\frac{dy}{dx}=\rho(x)^{1-\alpha}\propto x^{-\beta(1-\alpha)}$.
               
               Note that when $\alpha=\alpha_c=\frac{\beta-1}{\beta},$ then $-\beta\,(1-\alpha)=-1$, and so $y(x)\propto \ln x$, or equivalently $x(y)\propto e^y$. In this case we have that $\rho_\alpha(y)\propto[x(y)]^{-\beta\alpha}\propto e^{-(\beta-1) y}$. 
               
               Finally, when $\alpha> \frac{\beta-1}\beta$,  so $y(x)\propto x^{1-\beta(1-\alpha)}$; i.e, when $x,y>>1$ we have that $x(y)\propto y^{\frac1{1-\beta(1-\alpha)}}$. Thus,  $\rho_\alpha(y)\propto y^{\frac{-\beta\alpha}{1-\beta(1-\alpha)}}$. 
\end{proof}
It is interesting to note that under the conditions of Lemma \ref{lemma}, and in the high complexity limit, all the expected values become to be infinite, as well as the respective entropic moments $W_q$ when $q<1$. This is in concordance with the proposition \ref{prop7}, which states that, the entropic moments of the density are not well defined in the high complexity limit.

It is known that any distribution is characterized by its standard moments, provided that they exist. However, power-law-decaying probability densities does not fully satisfy this condition. In order to tackle this problem, Tsallis et al. \cite{Tsallis-Plastino09} purposed to use escort mean values. This make sense, taking into account that the escort density has more well defined moments than the original ones by choosing adequately the escort parameter. However, note that all escort transformation of a heavy tailed density remains being a heavy tailed, that is to say, a dense set of moments (with real parameter) remains always infinite. Contrary, as stated by Lemma \ref{lemma}, through the differential-escort density, we can always find a probability density which all its real moments correctly defined, at least for power-law-decaying probability densities. For these reasons, the characterization via \textit{diffierential escort} densities seems to be more accurate than via escort ones.
%\\
%On the other hand, the Lemma \ref{lemma} may be possibly extended to a more general case, taking into account that a \textit{nice} behavior of the density $\rho_{\alpha_c}$ is expectable, where 
%\begin{equation}
%\alpha_c=\sup\left\{\alpha\in\mathbb R \Bigg / \int_{\Lambda}[\rho(x)]^{1-\alpha}\,dx<\infty\right\}=\min\left\{\alpha\in\mathbb R \Bigg / \int_{\Lambda}[\rho(x)]^{1-\alpha}\,dx=\infty\right\}.
%\end{equation}
%More precisely, taking in mind that when $\alpha<\alpha_c$ the support of $\rho_\alpha$ is necessary compact, it seems reasonable conjecturing that $\rho_{\alpha_c}$ (whose support is not compact) has all their moments finite for any probability density $\rho$.  
%\textcolor{blue}{Note that Lemma \ref{lemma} implies that, for any density which decays as a power law it is always and unambiguously associated with and exponential decaying density, that is to say, all its moments are finite. This is a remarkable point, taking in mind that a numerable and infinite set of numbers (the moments) satisfying the Stieltjes and Hamburger conditions, fully characterizes a probability density, but nevertheless when not all moments are well defined this characterization is more difficult. So, through Lemma \ref{lemma}, adding the parameter $\alpha_c$ to the characterization of the moments, the statistical properties of the probability density adequately described..  TO BE WRITTEN... CITAR TSALLIS-PLASTINO-ALVAREZ-ESTRADA}

\section{Conclusions}\label{sec:conc}
In this paper we have presented the concept of differential-escort transformation of an univariate probability density. Its basic mathematical properties as composition and strong probability invariance have been studied. Then, we have  shown the regular behavior of the differential Shannon, Rényi and Tsallis  entropies for the differential-escort distributions. Moreover, the convex behavior of the Rényi entropy with respect to the differential-escort operation has been the keystone in the proof of  the monotonicity property of the LMC-Rényi complexity measure. Note that the differential-escort operation allows to define equivalence classes of probability densities where exists a total order with respect to their LMC-Rényi complexity. Later we have analyzed the statistical properties of a general probability density when it is  deformed to both extreme complexity cases, the low and high complexity limits. Finally, we have studied the behavior of the exponential and q-exponential densities, showing not only the stability of the q-exponential family, but also the existence of a critical value of the deformation parameter for what  the behavior of the tail, if any, dramatically changes to an exponential one. 

Interestingly, the action of this operation over a probability density allows for a clear interpretation of the probability conservation. Indeed, the conservation of the probability in any region of the transformed-space is clear by construction, what has a clear mass conservation interpretation.

On the other hand, the simplicity of the differential-escort transformations together with the general character of the presented results seem to indicate that this way of thinking would deserve to be explored from a more general point of view. Let us advance for example that the use of a differential-escort-based methodology has allowed for a huge generalization of the Stam inequality \cite{Zozor17}. 
\section*{Acknowledgement}

I am very grateful to Prof. J. S. Dehesa for useful discussions.

\appendix

\appendix

\end{document}